\documentclass[prl,twocolumn,aps,floatfix,showpacs,tightenlines,superscriptaddress,amsmath,amssymb,preprintnumbers,footinbib,10pt]{revtex4-2}

\usepackage{amsfonts,amssymb,amsbsy,epsfig,color,graphicx,bbold}
\usepackage[ansinew]{inputenc}
\usepackage[english]{babel}
\usepackage{braket}
\usepackage{comment}
\usepackage{ulem}
\usepackage{dcolumn}
\usepackage{hyperref}
\usepackage{enumitem}
\usepackage{amsmath,amsfonts,amssymb,amsthm}

\newcommand{\bea}{\begin{eqnarray}}
\newcommand{\eea}{\end{eqnarray}}
\newcommand{\be}{\begin{equation}}
\newcommand{\ee}{\end{equation}}

\newtheorem{thm}{Theorem}

\usepackage{titlesec}
\titleformat{\subsection}[runin]
{\normalfont\small\bfseries}{\thesubsection}{}{}

\begin{document}
	
\title{Quantum Phase Transition induced by Topological Frustration}

\author{Vanja Mari\'{c}}
\affiliation{Division of Theoretical Physics, Ru\dj{}er Bo\v{s}kovi\'{c} Institute, Bijeni\u{c}ka cesta 54, 10000 Zagreb, Croatia}
\affiliation{SISSA and INFN, via Bonomea 265, 34136 Trieste, Italy.}

\author{Salvatore Marco Giampaolo}
\affiliation{Division of Theoretical Physics, Ru\dj{}er Bo\v{s}kovi\'{c} Institute, Bijeni\u{c}ka cesta 54, 10000 Zagreb, Croatia}

\author{Fabio Franchini}
\affiliation{Division of Theoretical Physics, Ru\dj{}er Bo\v{s}kovi\'{c} Institute, Bijeni\u{c}ka cesta 54, 10000 Zagreb, Croatia}

\begin{abstract}
\textbf{Abstract.} In quantum many-body systems with local interactions, the effects of boundary conditions are considered to be negligible, at least for sufficiently large systems. Here we show an example of the opposite. We consider a spin chain with two competing interactions, set on a ring with an odd number of sites. When only the dominant interaction is antiferromagnetic, and thus induces topological frustration, the standard antiferromagnetic order (expressed by the magnetization) is destroyed. When also the second interaction turns from ferro to antiferro, an antiferromagnetic order characterized by a site-dependent magnetization which varies in space with an incommensurate pattern, emerges. This modulation results from a ground state degeneracy, which allows to break the translational invariance. The transition between the two cases is signaled by a discontinuity in the first derivative of the ground state energy and represents a quantum phase transition induced by a special choice of boundary conditions.
\end{abstract}

\preprint{RBI-ThPhys-2020-02}

\maketitle

\section{Introduction}

Modern physics follows a reductionist approach, in that it tries to explain a great variety of phenomena through the minimal amount of variables and concepts.
Thus, a successful theory should apply to a number as large as possible of situations and provide a predictive framework, depending on a number of variables as small as possible, within which one can describe the physical systems of interest.
On the other hand, further discoveries tend to enrich the phenomenology making more complicated, for the existing theories, to continue to predict accurately all the situations, sometimes to the point of exposing the need for new categories altogether.

Landau's theory of phases is a perfect example of such an evolution~\cite{Landau1978}.
Toward the middle of the last century~\cite{Landau1937}, all the different phases of many-body systems obeying classical mechanics were classified in terms of local order parameters that, turning from zero to a non-vanishing value, signal the onset of the corresponding order.
Each order parameter is uniquely associated with a particular kind of order, which in turn can be traced back to a specific local symmetry that is violated in that phase~\cite{Anderson1997}.
Hence symmetries play a key role in Landau's theory, while other features, such as boundary conditions, are deemed negligible (at least in the thermodynamic limit).

Because of its success, Landau's theory has been borrowed at first without modifications in the quantum regime~\cite{Sachdev2011}.
Nonetheless, after a few years, it has become clear that the richness of quantum many-body systems goes beyond the standard Landau paradigm.
Indeed, topologically ordered phases~\cite{Wen1989,Wen1990}, that have no equivalent in the classical regime, as well as nematic ones~\cite{Shannon2006}, represent instances in which violation of the same symmetry is associated with different (typically non-local) and non-equivalent order parameters~\cite{Lacroix2011,Giampaolo2015,Zonzo2018}, depending on the model under analysis.
This implied that Landau's theory had to be extended to incorporate more general concepts of order, which include the non-local effects that come along with the quantum regime and have no classical counterpart.

In more recent years, even boundary conditions, which are expected to be irrelevant for the onset of a classical ordered phase in the thermodynamic limit, have been shown to play a role when paired with quantum interactions.
Intuitively, one supposes that the contributions of boundary terms, that increase slowly with the size of the system with respect to the bulk ones, can be neglected when the dimension of the system diverges \cite{Burkhards1985,Cabrera1986,Cabrera1987}.
Recently, this intuition has been challenged. Thus in \cite{Campostrini2016} a concrete example of a boundary-driven quantum phase transition was provided, showing that, by tuning the coupling between the edges of an open chain, the system can visit different phases. In this line of research, particular attention was devoted to analyzing one-dimensional translational-invariant antiferromagnetic (AFM) spin models with frustrated boundary conditions (FBC), i.e, periodic boundary conditions in rings with an odd number of sites $N$.
For purely classical systems (Ising chains), FBC produce $2N$ degenerate lowest energy states, characterized by one domain wall defect in one of the two Neel orders.
Quantum effects split this degeneracy, producing, in the thermodynamic limit, a Galilean band of gapless excitations in touch with the lowest energy state(s)~\cite{Dong2016,Dong2017,Dong2018,Li2019} in a phase that, without frustration, would otherwise be gapped. In particular, while without frustration, the ground state of these models can be mapped exactly into the vacuum of a free fermionic system, the effect of FBC is to add a single excitation over this vacuum ~\cite{Giampaolo2019}. The na\"ive expectation is that, as the chain length is increased, the contributions from this single quasi-particle get diluted up to becoming irrelevant in thermodynamic limit. 
But this is not what was observed in \cite{Maric2019} where, in the presence of FBC, a short range dominant AFM interaction competes with a ferromagnetic one.
Indeed, the single-particle excitation brings $1/N$ corrections to the fundamental Majorana correlation functions, but these contributions can add up in the physical observables, due to the peculiar strongly correlated nature of the system. For instance, the two-point function, whose connected component is usually separated in the long distance limit to extract the spontaneous magnetization, acquires a multiplicative algebraic correction that suppresses it toward zero at distances scaling like the system size~\cite{Dong2016,Maric2019,MaricToeplitz}. The vanishing of the spontaneous magnetization and the replacement of the standard AFM local order with a mesoscopic ferromagnetic one was also established through the direct evaluation of the one point function in~\cite{Maric2019,MaricToeplitz}.

In the present work, we focus on the transition that occurs when also the second interaction becomes AFM.
This transition is characterized, even at finite size, by a level crossing associated with a discontinuity in the first derivative of the free energy at zero temperature (i.e., the ground state energy). In the phase where both interactions are AFM, the ground state becomes four-fold degenerate and this increased degeneracy allows for the existence of a different magnetic order.
This order is characterized by a staggered magnetization as in the standard AFM case, but with a modulation that makes its amplitude slowly varying in space. The results are surprising not only because of the order we find, but also because the quantum phase transition, signaled by the discontinuity, does not exist with other boundary conditions (BC), such as open (OBC) or periodic (PBC) boundary conditions with an even number of sites $N$. For this reason we term it "Boundary-conditions-induced Quantum Phase Transition" (BCI QPT).


\section{Results}

\subsection{Level crossing:} We illustrate our results by discussing the XY chain at zero field in FBC.
Even if this phenomenology is not limited to this model, it is useful to focus on it, because exploiting the well--known Jordan--Wigner transformation~\cite{Jordan1928} we can evaluate all the quantities that we need with an almost completely analytical approach.
The Hamiltonian describing this system reads
\begin{equation}\label{Hamiltonian}
    H=\sum\limits_{j=1}^N \cos\phi \ \sigma_j^x\sigma_{j+1}^x  +\sin\phi \ \sigma_j^y\sigma_{j+1}^y \; ,
\end{equation}
where $\sigma_j^\alpha$, with $\alpha=x,y,z$, are Pauli matrices and $N$ is the number of spins in the lattice.
Having assumed frustrated boundary conditions, we have that $N=2M+1$ is odd and $\sigma_j^\alpha\equiv \sigma_{j+N}^\alpha$.
The angle $\phi\in(-\frac{\pi}{4},\frac{\pi}{4})$ tunes the relative weight of the two interactions, as well as the sign of the smaller one.
Hence, while the role of the dominant term is always played by the AFM interaction along the $x$-direction, we have that the second Ising--like interaction switches from FM to AFM at $\phi=0$.

Regardless of the value of $\phi$, the Hamiltonian in eq.~\eqref{Hamiltonian} commutes with the parity operators $(\Pi^\alpha \equiv \otimes_{i=1}^N \sigma_i^\alpha)$, i.e. $[H,\Pi^\alpha]=0, \; \forall \alpha$.
At the same time, since we are considering odd $N$, different parity operators satisfy \mbox{$\left\{ \Pi^\alpha, \Pi^\beta \right\} = 2 \delta_{\alpha,\beta}$}, hence implying that each eigenstate is at least two-fold degenerate: if $\ket{\psi}$ is an eigenstate of both $H$ and $\Pi^z$, then $\Pi^x \ket{\psi}$, that differs from $\Pi^y \ket{\psi}$ by a global phase factor, is also an eigenstate of $H$ with the same energy but opposite $z$–parity. These symmetries are important because they imply an exact ground-state degeneracy even in finite chains and thus the possibility to select states with a definite magnetization within the ground state manifold (for more details about the symmetries of the model see Supplementary Note 1). Furthermore, using the techniques introduced in~\cite{Maric2019}, it is possible to directly evaluate the magnetization of these states: having it as a function of the number of sites of the chain, we can take the thermodynamic limit and thus recover directly its macroscopic value, without resorting to the usual approach making use of the cluster decomposition.

\begin{figure}
	\includegraphics[width=0.9\columnwidth]{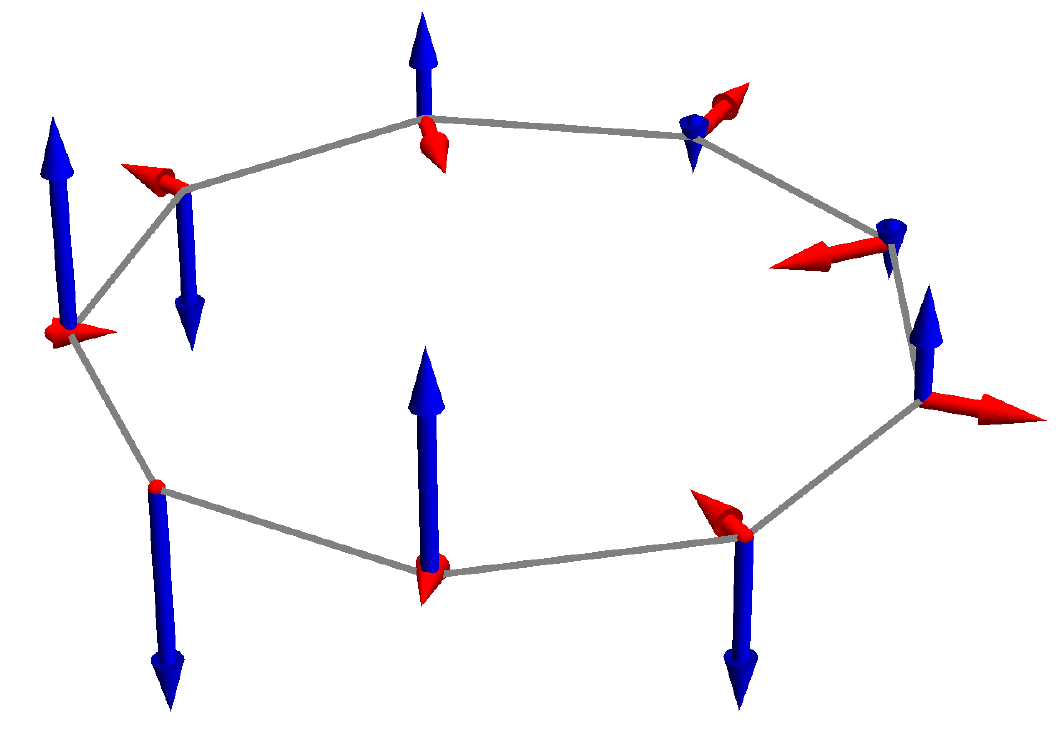}
	\caption{\textbf{Three-dimensional representation of the
			site-dependent magnetization.} (Color online) Site dependent magnetizations along $x$ (Blue darker arrows) and $y$ (Red lighter arrows) for each spin of a lattice with $N=9$ sites. The magnetizations are obtained setting $\phi=\frac{\pi}{8}$ and recovering the maximum amplitudes $f_x\simeq0.613$ and $f_y\simeq0.329$, see discussion around eq.~\eqref{spatial}.}
	\label{plots}
\end{figure}

Using the standard techniques ~\cite{Franchini2017}, that consist in the Jordan-Wigner transformation and a Fourier transform followed by a Bogoliubov rotation (more details in Supplementary Note 2), the Hamiltonian can be reduced to
\begin{eqnarray}
\label{supp_Hamiltonian_2}
H&=&\frac{1+\Pi^z}{2}H^+ \frac{1+\Pi^z}{2} + \frac{1-\Pi^z}{2}H^- \frac{1-\Pi^z}{2} \;, \nonumber \\
H^\pm&=&\sum\limits_{q\in\Gamma^\pm}^{} \varepsilon(q) \left(a_q^\dagger a_q-\frac{1}{2}\right)\,.
\end{eqnarray}
Here $a_q$ ($a_q^\dagger$) is the annihilation (creation) fermionic operator with momentum $q$. The Hilbert space has been divided into the two sectors of different $z$-parity $\Pi^z$. Accordingly, the momenta run over two disjoint sets, corresponding to the two sector: \mbox{$\Gamma^-=\{2\pi k/N \}$} and $\Gamma^+=\{2\pi (k+\frac{1}{2})/N \}$ with $k$ ranging over all integers from $0$ to $N-1$.
The dispersion relation reads
\begin{eqnarray}
\label{energy_momenta}
\epsilon(q) &=& 2 \left| \cos\phi \ e^{\imath 2 q} + \sin\phi \right| ,\  q\neq 0, \pi \ , \nonumber \\
\epsilon(0) &=&- \epsilon(\pi)=2 \left( \cos\phi+\sin\phi \right) \; ,
\end{eqnarray}
where we note that only $\epsilon(0)$, $\epsilon(\pi)$ can become negative.
	
The eigenstates of $H$ are constructed by populating the vacuum states $\ket{0^\pm}$ in the two sectors and by taking care of the parity constraints. The effect of frustration is that the lowest energy states are not admissible due to the parity requirement. For instance, from {eq.~\eqref{energy_momenta} we see that, assuming $\phi\in(-\frac{\pi}{4},\frac{\pi}{4})$, the single negative energy mode is $\epsilon(\pi)$, which lives in the even sector $(\pi\in\Gamma^+)$.
Therefore the lowest energy states are, respectively, $\ket{0^-}$ in the odd sector and $a^\dagger_\pi \ket{0^+}$ in the even one.
But, since both of them violate the parity constraint of the relative sector, they cannot represent physical states.
Hence, the physical ground states must be recovered from $\ket{0^-}$ and $a^\dagger_\pi \ket{0^+}$ considering the minimal excitation coherent with the parity constraint.

While for $\phi<0$ there is a unique state in each parity sector that minimizes the energy while respecting the parity constraint (and these states both have zero momentum), for $\phi>0$ the dispersion relation in eq.~\eqref{energy_momenta} becomes a double well and thus develops two minima: $\pm p \in\Gamma^-$ and $\pm p'\in\Gamma^+$, approximately at $\pi/2$ (for their precise values and more details, see ``{\it Methods}'').
Thus, for $\phi>0$ the ground state manifold becomes $4$-fold degenerate, with states of opposite parity and momenta. This degeneracy has a solid geometrical origin, which goes beyond the exact solution to which the XY is amenable, and has to do with the fact that, with FBC, the lattice translation operator does not commute with the mirror (or chiral) symmetry, except than for states with $0$ or $\pi$ momentum (see Supplementary Note 4). Thus, every other state must come in degenerate doublets of opposite momentum/chirality.
In accordance to this picture, a generic element in the four-dimensional ground state subspace can be written as
\be
\label{ground_state_superposition1}
\ket{g}  = u_1\ket{p}+u_2\ket{-p}+u_3 \ket{p'}+u_4 \ket{-p'} \ , 
\ee
where the superposition parameters satisfy the normalization constraint $\sum_i |u_i|^2=1$, $\ket{\pm p}\!=\!a_{\pm p}^\dagger\!\ket{0^-}$ are states in the odd $z$-parity sector and ${\ket{\pm p'}=\Pi^x\ket{\mp p}=a_{\pm p'}^\dagger a_\pi^\dagger\!\ket{0^+}}$ are the states in the even sector (for the second equality, that holds up to a phase factor, see ``{\it Methods}'').
	
Hence, independently from $N$, once FBC are imposed, the system presents a level crossing at the point $\phi=0$, where the Hamiltonian reduces to the classical AFM Ising.
The presence of the level crossing is reflected on the behavior of the ground state energy $E_{g}$, whose first derivative exhibits a discontinuity
\begin{equation}
\label{derivative}
\frac{dE_g}{d\phi}\bigg|_{\phi\to0^-}-\frac{dE_g}{d\phi}\bigg|_{\phi\to0^+}=2\Big(1+\cos\frac{\pi}{N}\Big) ,
\end{equation}
which goes to a nonzero finite value in the thermodynamic limit.
The presence of both a discontinuity in the first derivative of the ground state energy, and a different degree of degeneracy even at finite sizes, is coherent with a first-order quantum phase transition~\cite{Sachdev2011}.

However, such a transition is present only when FBC are considered.
Indeed, without frustration, hence considering either OPC or PBC conditions in a system with even $N$, the two regions $\phi\in(-\frac{\pi}{4},0)$ and $\phi\in(0,\frac{\pi}{4})$ belong to the same AFM phase, have the same degree of ground-state degeneracy, and exhibit the same physical properties~\cite{Lieb1961,Barouch1971}.
Hence it is the introduction of the FBC that induces the presence of a quantum phase transition at $\phi=0$.

\subsection{The magnetization:}

\begin{figure}
	\includegraphics[width=0.9\columnwidth]{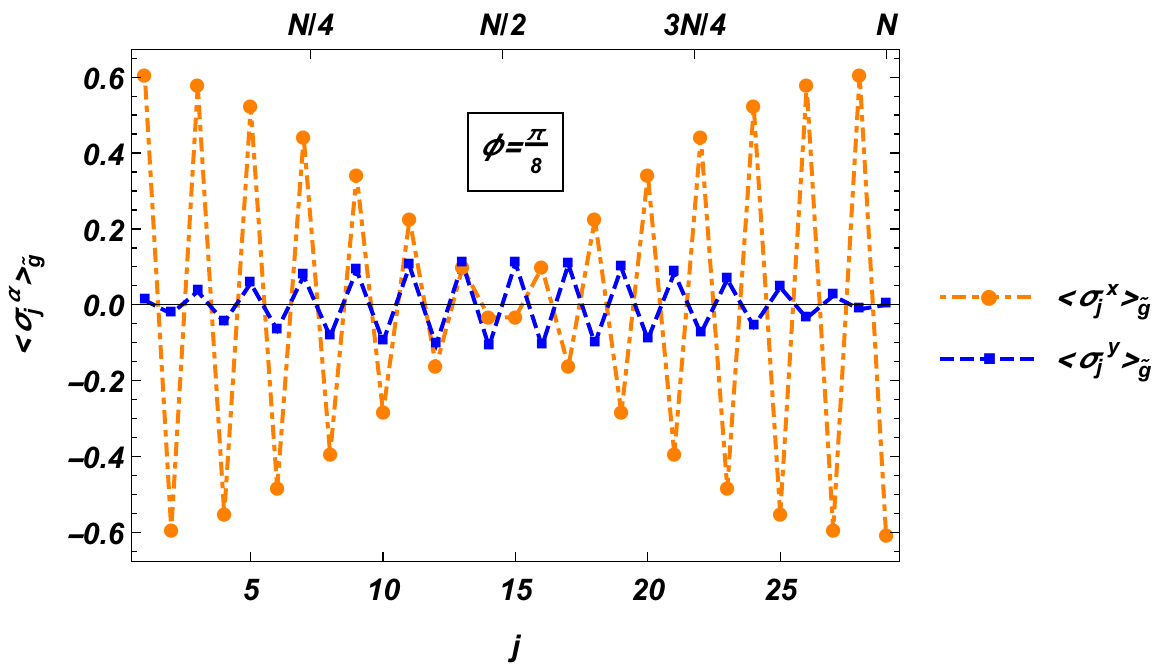}
	\caption{\textbf{Site-dependent magnetization.} (Color online) Plot of the site dependent magnetizations along $x$ (orange points) and $y$ (blue ones) for each spin of a lattice with $N=29$ sites. The magnetizations are obtained setting $\phi=\frac{\pi}{8}$. The dashed lines are a guide to the eye to show the almost staggered order, while the modulation in space is given by eq.~\ref{spatial}.}
	\label{2dplots}
\end{figure}

Having detected a phase transition, we need to identify the two phases separated by it.
In~\cite{Maric2019} it was proved that the two-fold degenerate ground state for $\phi<0$ is characterized by a ferromagnetic mesoscopic order: for any finite odd $N$, the chain exhibits non-vanishing, site-independent, ferromagnetic magnetizations along any spin directions. These magnetizations scale proportionally to the inverse of the system size and, consequently, vanish in the thermodynamic limit.
For suitable choices of the ground state, this mesoscopic magnetic order is present also for $\phi>0$ but, taking into account that now the ground state degeneracy is doubled, this phase can also show a different magnetic order, that is forbidden for $\phi<0$.
However, from all the possible orders that can be realized we can, for sure, discard the standard staggerization that characterizes the AFM order in the absence of FBC.
In fact, for odd $N$, it is not possible to align the spins perfectly antiferromagnetically, while still satisfying PBC.
In a classical system, the chain develops a ferromagnetic defect (a domain wall) at some point, but quantum-mechanically this defect gets delocalized  and its effect is not negligible in the thermodynamic limit as one would naively think.

To study the magnetization let us consider a ground state vector that is not an eigenstate of the translation operator:
\begin{equation}
\label{smartgs}
\ket{\tilde{g}}=\frac{1}{\sqrt{2}}\big( \ket{p} +e^{\imath \theta} \ket{p'}\big) \; ,
\end{equation}
where $\theta$ is a free phase. We compute the expectation value of spin operators on this state. Having broken translational invariance, we can expect the magnetization to develop a site dependence, which can be found by exploiting the translation and the mirror symmetry (see ``{\it Methods}''), giving
\begin{eqnarray}\label{spatial}
\langle \sigma_j^\alpha \rangle_{\tilde{g}} \! &\!=\! &\! (-1)^{j}\! \cos\!\left[ \pi\frac{j}{N}+ \lambda(\alpha,\theta,N)\right] \! f_\alpha \ , \;\;\;\;\;
\end{eqnarray}
where $f_\alpha \equiv |\!\bra{p}\sigma_N^\alpha\ket{p'}\!|$. The two phase factors, whose explicit dependence on the arbitrary phase $\theta$ is given in Supplementary Note 5, are related as $\lambda(y,\theta,N)-\lambda(x,\theta,N)=\pi/2$, which corresponds to a shift by half of the whole ring between the $x$ and $y$ magnetization profiles. The obtained spatial dependence, depicted in Figure~\ref{plots} and \ref{2dplots}, thus breaks lattice translational symmetry, not to a reduced symmetry as in the case of the staggerization that characterizes the standard AFM order, but completely, since we have an incommensurate modulation that depends on the system size over-imposed to the staggerization. 

While the simple argument just presented explains how and why the magnetizations along $x$ and $y$ acquire a nontrivial spatial dependence, we still have to determine how their magnitudes scale with $N$. The magnitudes depend on the spin operator matrix elements $\bra{p}\sigma_N^\alpha\ket{p'}$ and their evaluation is explained in ``{\it Methods}''.

\begin{figure}
    \includegraphics[width=1.0\columnwidth]{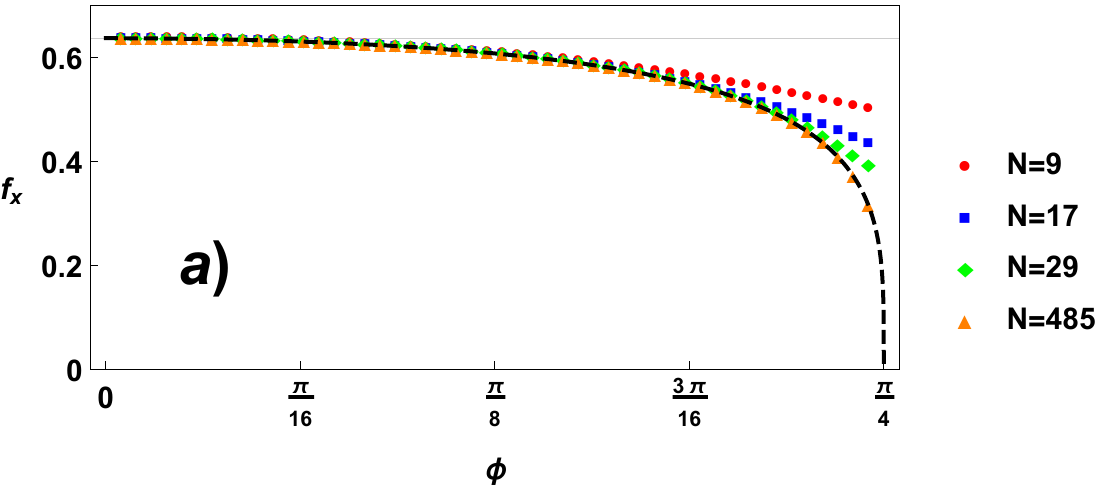}
    \includegraphics[width=1.0\columnwidth]{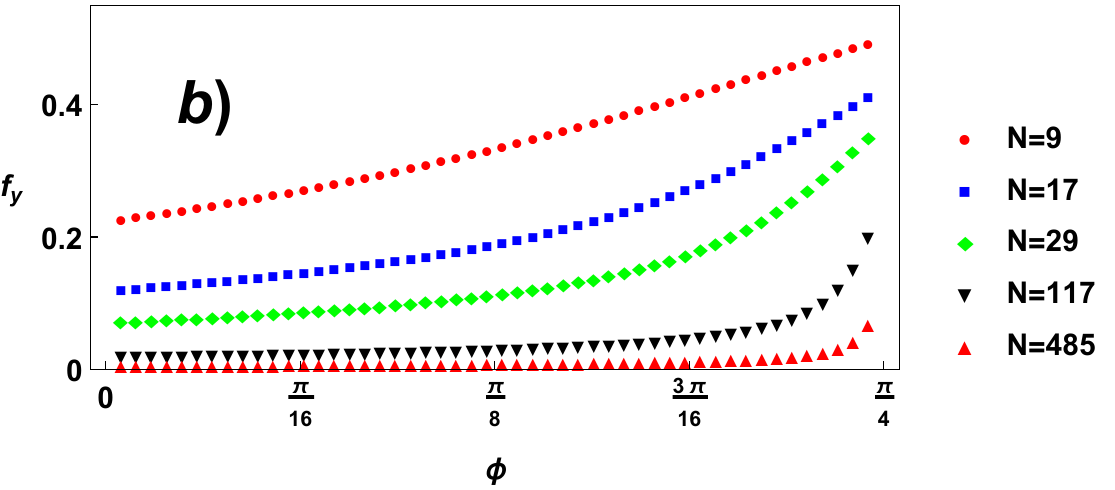}
    \caption{\textbf{Matrix elements that determine the magnetization.} (Color online) Behavior of matrix elements $f_x$ (\textbf{a}) and $f_y$ (\textbf{b}) as function of the Hamiltonian parameter $\phi$ for different sizes of the the system $N$. The magnetizations are site-dependent, as given by the formula $\langle \sigma_j^\alpha \rangle_{\tilde{g}} = (-1)^{j} \cos\left[ \pi\frac{j}{N}+ \lambda(\alpha,\theta,N)\right]  f_\alpha  $ for $\alpha=x,y$, where $\lambda$ is a phase factor that depends on additional details of the ground state. The matrix elements $f_x$ and $f_y$ thus determine the maximal value the magnetization can achieve over the ring.}
    \label{fig2}
\end{figure}
As we can see from Figure~\ref{fig2}, we have two different behaviors for the magnetizations along $x$ and $y$.
While for the former we can see that it admits a finite non zero limit, which is a function of the parameter $\phi>0$, the latter, for large enough systems, is proportional to $1/N$ (see also Figure~\ref{fig3}) and vanishes in the thermodynamic limit.
Hence, differently from the one along the $y$ spin direction, the {\it ``incommensurate antiferromagnetic order''} along $x$ survives also in the thermodynamic limit. By exploiting perturbative analysis around the classical point $\phi=0$ it is possible to show that, for $\phi\to0^+$ and diverging $N$, $f_x$ goes to $2/\pi$ (see Supplementary Note 7 for details). 
Moreover, numerical analysis has also shown that in the whole region $\phi\in(0,\pi/4)$ we have
\begin{equation}
\lim_{N\to\infty} |\bra{p}\sigma_N^x\ket{p'}|=\frac{2}{\pi}(1-\tan^2\phi)^{\frac{1}{4}} \ .
\end{equation}

\section{Discussion}

\begin{figure}
	\includegraphics[width=1\columnwidth]{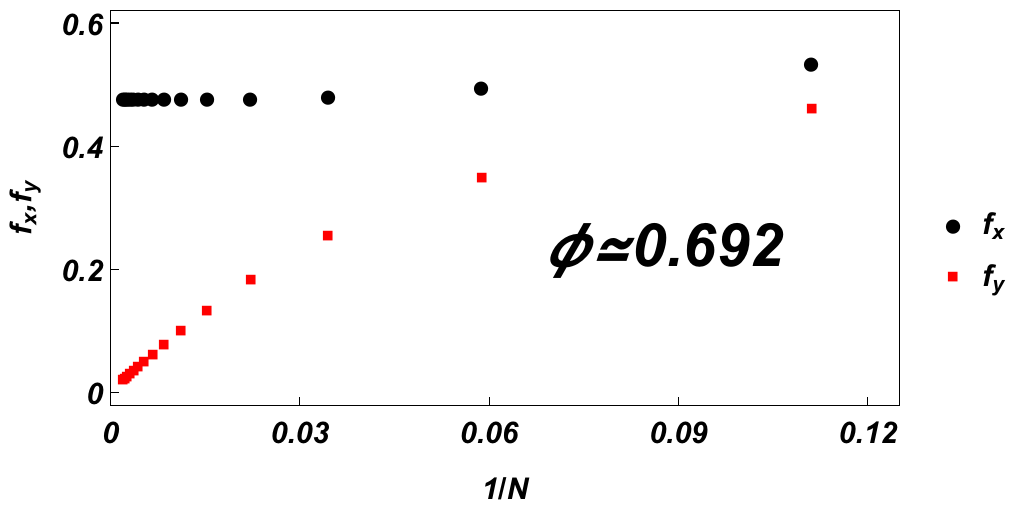}
	\caption{\textbf{Difference in the scaling of the two matrix elements.} Dependence of the two $f_\alpha=|\bra{p}\sigma_N^\alpha\ket{p'}|$ on the inverse of the size of the system $1/N$ for $\phi\simeq0.692$. The black points represent the values obtained for $f_x$ while the red squares stand for $f_y$.}
	\label{fig3}
\end{figure}

Summarizing, we have proved how, in the presence of FBC, the Hamiltonian in eq.~\eqref{Hamiltonian} shows a quantum phase transition for $\phi=0$.
Such transition is absent both for OBC and for systems with PBC made of an even number of spins.
This quantum phase transition separates two different gapless, non-relativistic phases that, even at a finite size, are characterized by different values of ground-states degeneracy: one shows a two-fold degenerate ground-state, while in the second we have a four-fold degenerate one.
This difference, together with the fact that the first derivative of the ground-state energy shows a discontinuity in correspondence with the change of degeneracy,
supports the idea that there is a first-order transition.

The two phases display the two ways in which the system can adjust to the conflict between the local AFM interaction and the global FBC: either by displaying a mesoscopic ferromagnetism, whose magnitude decays to zero with the system size~\cite{Maric2019}, or through an approximate staggerization, so that the phase difference between neighboring spins is $\pi \left(1 \pm \frac{1}{N} \right)$.
For large systems, these $1/N$ corrections induced by frustration are indeed negligible at short distances.
However, they become relevant when fractions of the whole chain are considered. 
Crucially, the latter order spontaneously breaks translational invariance and remains finite in the thermodynamic limit. 
Let us remark once more that, with different boundary conditions, all these effects are not present.

The results presented in this work are much more than an extension of~\cite{Maric2019}, in which we already proved that FBC can affect local order. 
While in~\cite{Maric2019} AFM was destroyed by FBC and replaced with a mesoscopic ferromagnetic order, here we encounter an AFM order, which spontaneously breaks translational invariance, is modulated in an incommensurate way, and does not vanish in the thermodynamic limit. Most of all, the transition between these two orders is signaled by a discontinuity in the derivative of the free energy, indicating a first-order quantum phase transition.

The phase transition we have found resembles several well-known phenomena of quantum complex systems, without being completely included in any of them. A finite difference of the values of the free energy derivative at two sides of the transition characterizes also first-order wetting transitions~\cite{Diehl1986,Bonn2001,Bonn2009}, that are associated to the existence of a border. On the other hand, in our system, we cannot individuate any border, since the chain under analysis is perfectly invariant under spatial translations. Delocalized boundary transitions have already been reported and are called “interfacial wetting”, but they differ from the phenomenology we discussed here, as they refer to multi-kink states connecting two different orders (prescribed at the boundary) separated by a third intermediate state \cite{Delfino2016}.

The transition we have found, and the incommensurate AFM order, might also be explored experimentally. To observe them, one could, for example, measure the magnetization at different positions in the ring. In the phase exhibiting incommensurate AFM order, the measurements will yield different values at different positions, while in the other phase, exhibiting mesoscopic ferromagnetic order the values are going to be the same. One could also examine the maximum value of the magnetization over the ring. In the incommensurate AFM phase this value is finite, while in the other it goes to zero in the thermodynamic limit. The maximum of the magnetization over the ring thus exhibits a jump at the transition point.

The strong dependence of the macroscopic behavior on boundary conditions that we have found seemingly contradicts one of the tenants of Landau Theory and we cannot offer at the moment a unifying picture that would reconcile our results with the general theory. 
Indeed, FBC are special, as the kind of spin chains we consider are the building blocks of every frustrated system~\cite{Toulouse1977,Vannimenus1977,Wolf2003,Giampaolo2011,Marzolino2013,Giampaolo2015_2}
, which are known to present peculiar properties. We can also speculate that FBC induce a topological effect that puts the system outside the range of validity of Landau's theory. In fact, while in the ferromagnetic phases of the model the ground state degeneracy in the thermodynamic limit is independent of boundary conditions, in the parameter region exhibiting incommensurate AFM order the degeneracy is doubled with FBC, thus clearly depending on the (real space) topology of the system.  But, there is a second more subtle connection. 
Indeed, while magnetic phases show symmetry-breaking order parameters, topological phases are characterized by the expectation value of a non-local string operator that does not violate the bulk symmetry of the system. In our system, as we have shown before, the value of the local magnetization is associated with the expectation value of the operator $\sigma_N^x\Pi^x=\bigotimes_{j=1}^{N-1} \sigma_j^x$,  which is a string operator that does not break the parity symmetries of the model. However, while geometrical frustration induces some topological effects in the XY chain, interestingly, we have found evidence that suggests that topological phases are resilient to geometrical frustration \cite{Maric20_3}.

A natural question that emerges is how robust is the observed phenomenology to defects, that destroy the translational symmetry of the model. In fact, a common expectation is that such defect would pin the domain wall and restore the unfrustrated physics in the bulk. This question has been addressed in~\cite{Torre2020}, where it has been shown that a complex picture emerges depending on the nature of the defects, but that ultimately the incommensurate AFM order can survive under very general conditions. Thus, the physics we have discussed in this work is not only a remarkable point of principle but also a physically measurable phenomenon.

\section{Methods}

\subsection{Ground state degeneracy:}

We have two different pictures depending on the sign of $\phi$.
For $\phi<0$ the excitation energy, given  by eq.~\eqref{energy_momenta}, admits two equivalent local minima, one for each parity, i.e.  \mbox{$q=0\in\Gamma^-$} and $q=\pi\in\Gamma^+$.
Consequently, the ground state is two-fold degenerate, and the two ground states that are also eigenstates of $\Pi^z$ are $\ket{g_0^-}=a^\dagger_0 \ket{0^-}$ and $\ket{g_0^+}=\Pi^x \ket{g_0^-}=\ket{0^+}$, where the last equality holds up to a phase factor. On the contrary, when $\phi$ becomes positive, the energy in eq.~\eqref{energy_momenta} admits, for each $z$-parity sector, two local minima at opposite momenta, $\pm p \in\Gamma^-$ and $\pm p'\in\Gamma^+$, where $p=\frac{\pi}{2} \left( 1 - \frac{1}{N} \right)$ for a system size $N$ satisfying $N\textrm{ mod }4=1$, $p=\frac{\pi}{2} \left( 1 + \frac{1}{N} \right)$ for $N\textrm{ mod }4=3$ and $p'=\pi-p$.

\subsection{Spatial dependence of the magnetization:}
To study the spatial dependence of the magnetization it is useful to introduce the unitary lattice translation operator $T$, whose action shifts all the spins by one position in the lattice as
\begin{equation}\label{T spins}
T^\dagger\sigma_j^\alpha T=\sigma_{j+1}^\alpha , \quad \alpha=x,y,z
\end{equation}
and which commutes with the system's Hamiltonian in eq.~\eqref{Hamiltonian}, i.e. $[H,T]\!=\!0$.
The operator $T$ admits, as a generator, the momentum operator $P$, i.e. $T=e^{\imath P}$. Among the eigenstates of $P$, we have the ground state vectors  $\ket{\pm p}$ and  $\ket{\pm p'}$ with relative eigenvalues equal to $\pm p$ and $\pi\pm p'=\mp p$. A detailed definition of the operator and a proof of these properties is given in Supplementary Note 3. The latter equality allows to identify the ground states $a_{\pm p'}^\dagger a_\pi^\dagger\!\ket{0^+}$ with the states $\Pi^x\ket{\mp p}$. 

We can exploit the properties of the operator $T$ to determine, for each odd $N$, the spatial dependence of the magnetizations along $x$ and $y$ in the ground state $\ket{\tilde{g}}$ ($\langle \sigma_j^\alpha \rangle_{\tilde{g}}$ with $\alpha=x,\,y$), defined in eq.~\eqref{smartgs}. In fact, taking into account that $\ket{p}$ and $\ket{p'}$ live in two different $z$-parity sectors, we have that the magnetization along a direction orthogonal to $z$ on the state $\ket{\tilde{g}}$ is given by 
\begin{equation}\label{step magnetization first}
\!\!\langle \sigma_j^\alpha \rangle_{\tilde{g}}\!=\!\bra{\tilde{g}}\sigma_j^\alpha\ket{\tilde{g}}\!=\!\frac{1}{2}\big(  e^{i\theta}\bra{p}\sigma_j^\alpha\ket{p'}\!+ e^{-i\theta}\bra{p'}\sigma_j^\alpha\ket{p}\big).
\end{equation}
The magnetization is determined by the spin operator matrix elements $\bra{p}\sigma_j^\alpha\ket{p'}$, that can all be related to the ones at the site $j=N$. In fact, considering eq.~\eqref{T spins} we obtain
\begin{equation}\label{matrix element site j}
\bra{p}\sigma_j^\alpha\ket{p'}=e^{-\imath 2 p j}\bra{p}\sigma_N^\alpha\ket{p'} \ .
\end{equation}

The advantage of this representation is that the matrix element $\bra{p}\sigma_N^\alpha\ket{p'}$ is a real number for $\alpha=x$, and a purely imaginary one for $\alpha=y$, making it simple to express the magnetization. Let us illustrate the computation of the $x$ magnetization, while the details for the $y$ magnetization can be found in Supplementary Note 5. The special role of the site $N$ is singled out by the choice made in the construction of the states through the Jordan-Wigner transformation. To prove that the matrix element is real it is useful to introduce the, unitary and hermitian, mirror operator with respect to site $N$, denoted as $M_N$, that makes the mirroring
\begin{equation}
M_N\sigma_j^\alpha M_N=\sigma_{-j}^\alpha, \quad \alpha=x,y,z,
\end{equation}
and, in particular, leaves the $N$-th site unchanged. The operator satisfies $M_N\ket{\pm p}=\ket{\mp p}$, while the reflections with respect to other sites would introduce additional phase factors. A detailed definition of the mirror operators and discussion of their properties is given in Supplementary Note 4. Exploiting the properties of $M_N$ we have then
\begin{equation}
\bra{p}\sigma_N^x\Pi^x \ket{-p}=\bra{-p}\sigma_N^x\Pi^x\ket{p}=(\bra{p}\sigma_N^x\Pi^x\ket{-p})^*,
\end{equation}
so $\bra{p}\sigma_N^x\ket{p'}$ is real. Evaluating $\bra{p}\sigma_N^x\ket{p'}$ using the methods of the next paragraph we can see that the quantity is actually positive, and therefore equal to its magnitude $f_x$. Then from eq.~\eqref{step magnetization first} and \eqref{matrix element site j} we get the spatial dependence of the magnetization 
\begin{equation}
\braket{\sigma_j^x}_{\tilde{g}}=\cos(2pj-\theta)\bra{p}\sigma_N^x\ket{p'}.
\end{equation}
Inserting the exact value of the momentum we get eq.~\eqref{spatial} for $\alpha=x$, where the exact value of $\lambda(x,\theta,N)$ is given in Supplementary Note 5.

\subsection{Scaling of the magnetization with \textit{N}:}
The magnetization is determined by the matrix elements $f_\alpha=|\bra{p}\sigma_N^\alpha\ket{p'}|$. To evaluate them we exploit the trick introduced in~\cite{Maric2019} and used to compute the magnetization.

Within the ground state manifold, we define the vectors
\begin{equation}
\label{smartgs2}
\ket{g_\pm} \equiv \frac{1}{\sqrt{2}}(\ket{p} \pm \ket{-p}  ) ,
\end{equation}
and, further using the, already introduced, properties of the mirror operator $M_N$ (see Supplementary Note 6 for details), we get
\begin{eqnarray}
\!\bra{p}\sigma_N^\alpha\ket{p'} \!&\! =\! & \frac{1}{2}\Big(\! \bra{g_+}\sigma_N^\alpha\Pi^x\ket{g_+}\!-\!\bra{g_-}\sigma_N^\alpha\Pi^x\ket{g_-}\!\Big) \ .
\nonumber \\
\end{eqnarray}
In this way, we represent a notoriously hard one point function in terms of standard expectation values of products of an even number of spin operators $\sigma_N^\alpha\Pi^x$, which can be expressed as a product of an even number (parity preserving) of fermionic operators. Using Wick's theorem, the expectation values can then be expressed as determinants and evaluated numerically efficiently (see Supplementary Note 6).

Moreover, in the limit $\phi\to 0^+$ the matrix elements can also be evaluated analytically using a perturbative approach (see Supplementary Note 7).

\section*{Acknowledgments}
We thank Giuseppe Mussardo, Rosario Fazio, and Marcello Dalmonte for useful discussions and suggestions.
We acknowledge support from the European Regional Development Fund -- the Competitiveness and Cohesion Operational Programme (KK.01.1.1.06 -- RBI TWIN SIN) and from the Croatian Science Foundation (HrZZ) Projects No. IP--2016--6--3347 and IP--2019--4--3321.
SMG and FF also acknowledge support from the QuantiXLie Center of Excellence, a project co--financed by the Croatian Government and European Union through the European Regional Development Fund -- the Competitiveness and Cohesion (Grant KK.01.1.1.01.0004).

\appendix

\begin{widetext}

\section{The model and its symmetries}

The XY chain studied in the letter is given by the Hamiltonian
\begin{equation}\label{Hamiltonian}
H=\sum\limits_{j=1}^N \Big( \cos\phi \ \sigma_j^x\sigma_{j+1}^x  +\sin\phi \ \sigma_j^y\sigma_{j+1}^y\Big) \; ,
\end{equation}
where $\sigma_j^\alpha$, with $\alpha=x,y,z$, are Pauli operators acting on the $j$-th spin, $N$ is the number lattice sites and we assume frustrated boundary conditions (FBC), given by periodic boundary conditions $\sigma_j^\alpha=\sigma_{j+N}^\alpha$ and an odd number of lattice sites. 
In these supplementary materials we will focus on the region $\phi\in(0,\pi/4)$, where both of the two interactions are antiferromagnetic.
We also compare the results obtained for this region with the one analyzed in Ref.~\cite{Maric2019}, which, keeping $\phi\in(-\pi/4,0)$, describes the situation where one dominant, antiferromagnetic coupling appears together with a ferromagnetic smaller one.

Since the model in eq.~\eqref{Hamiltonian} does not include an external magnetic field, the Hamiltonian 
commutes with all three parity operators $\Pi^\alpha \equiv \bigotimes_{j=1}^{N} \sigma_j^\alpha, \ \alpha=x,y,z$, i.e. $[H,\Pi^\alpha]=0 ,\ \forall \alpha$.
However, assuming FBC and hence setting the number of sites to be an odd number, different parity operators anticommute, satisfying $ \left\{ \Pi^\alpha, \Pi^\beta \right\} = 2 \delta_{\alpha,\beta}$.
The fact that the different parity operators anticommute has an immediate relevant consequence: each eigenstate is at least two-fold degenerate. 
To explain this point, let us assume that $\ket{\varphi}$ is simultaneously an eigenstate of $H$ and one of the three parity operators, for instance $\Pi^z$. 
Then, the image of $\ket{\varphi}$ under the action of one of the other parity operators, for example $\Pi^x\ket{\varphi}$, is still an eigenstate of both $H$ and $\Pi^z$.
But while $\ket{\varphi}$ and $\Pi^x\ket{\varphi}$ have the same energy, they have different $z$ parity.
As a consequence, for each eigenstate of the Hamiltonian in the even sector of one of the parities ($\Pi^\alpha=1$), there will be a second eigenstate of the Hamiltonian, with the same energy but living in the odd sector ($\Pi^\alpha=-1$).
Hence each eigenvalue of the Hamiltonian is, at least, two-fold degenerate.

However, other symmetry properties of the Hamiltonian will prove to be of extreme relevance in the following.
At first, due to periodic boundary conditions, the model exhibits exact translational symmetry, which is expressed in the commutation of the Hamiltonian with the lattice translation operator $T$. 
Finally, the model also exhibits mirror symmetry with respect to any lattice site. Namely, for any lattice site $k$ the Hamiltonian in eq.~\eqref{Hamiltonian} is invariant under the mirror image with respect to it, achieved by the transformation $j\to 2k-j$ on spins, associated to the action of the mirror operator $M_k$.

\section{Exact solution}

As it is well-known, the model in eq.~\eqref{Hamiltonian} can be diagonalized exactly, using standard techniques of mapping spins to fermions~\cite{Franchini2017}. 
The Jordan-Wigner transformation defines the fermionic operators as
\begin{equation}
c_j=\Big(\bigotimes\limits_{l=1}^{j-1}\sigma_l^z\Big)\otimes\sigma_{j}^+  , \quad c_j^\dagger=\Big(\bigotimes\limits_{l=1}^{j-1}\sigma_l^z\Big)\otimes\sigma_{j}^-,
\end{equation}
where $\sigma_j^\pm=(\sigma_j^x\pm \imath\sigma_j^y)/2$ are spin raising and lowering operators. In this notation, not explicitly mentioning a lattice site in the tensor product corresponds to making a tensor product with an identity operator on that site. In terms of Jordan-Wigner fermionic operators, the Hamiltonian in eq.~\eqref{Hamiltonian} reads as
\begin{eqnarray}
\label{supp_Hamiltonian_1}
H\!&\!=\!&\! \sum\limits_{j=1}^{N-1}\!\big[ (\sin\phi-\cos\phi) c_jc_{j+1}\!-\!
(\cos\phi+\sin\phi)c_jc_{j+1}^\dagger +\textrm{h.c.} \big]\! -\! \Pi^z \big[(\sin\phi-\cos\phi) c_Nc_{1}\!-\!
(\cos\phi+\sin\phi)c_Nc_{1}^\dagger+\textrm{h.c.} \big]
\end{eqnarray}
Due to the presence of the parity operator along $z$, the Hamiltonian given by eq.~\eqref{supp_Hamiltonian_1} is not in a quadratic form, but becomes quadratic in each of the two parity sector of $\Pi^z$, i.e.
\begin{equation}
\label{supp_Hamiltonian_2}
H=\frac{1+\Pi^z}{2}H^+ \frac{1+\Pi^z}{2} + \frac{1-\Pi^z}{2}H^- \frac{1-\Pi^z}{2} \; ,
\end{equation}
where both $H^+$ and $H^-$ are quadratic. 
Being quadratic, they can be brought to a form of free fermions, which is done conveniently in two steps. First, $H^\pm$ are written in terms of the Fourier transformed Jordan-Wigner fermions,
\begin{equation}\label{fourier transformed JW fermions}
b_q=\frac{1}{\sqrt{N}}\sum\limits_{j=1}^{N} c_j \ e^{-\imath qj} , \quad b_q^\dagger=\frac{1}{\sqrt{N}}\sum\limits_{j=1}^{N} c_j^\dagger \ e^{\imath qj}  ,
\end{equation}
for $q\in\Gamma^\pm$, where the two sets of quasi-momenta are given by $\Gamma^-=\{2\pi k/N \}$ and $\Gamma^+=\{2\pi (k+\frac{1}{2})/N \}$ with $k$ running on all integers between $0$ and $N-1$. 
Then a Bogoliubov rotation
\begin{equation}\label{Bogoliubov particles}
\begin{split}
&a_q=\cos\theta_q \ b_q + \imath \sin\theta_q \ b_{-q}^\dagger, \quad q\neq0,\pi\\
&a_{q}=b_q, \quad q=0,\pi 
\end{split}
\end{equation}
with a momentum-dependent Bogoliubov angle given by
\begin{equation}\label{arctan}
\theta_{q}\!=\!\arctan \frac{|\sin \phi + \cos \phi \ e^{\imath 2q}|\! -\!(\sin\phi+\cos\phi)\cos q}{(\cos\phi-\sin\phi)\sin q}
\end{equation}
is used to bring them to a form of free fermions. 
We end up with
\begin{equation}
\label{supp_Hamiltonian_3}
H^\pm=\sum\limits_{q\in\Gamma^\pm}^{} \varepsilon(q) \left(a_q^\dagger a_q-\frac{1}{2}\right) ,
\end{equation}
where the dispersion law is given by
\begin{eqnarray}
\label{supp_energy_momenta}
\epsilon(q) &=&2\left| \sin\phi+\cos\phi \ e^{\imath 2q} \right| ,\  q\neq 0, \pi \ , \nonumber \\
\epsilon(0) &=&-\epsilon(\pi)=2(\sin\phi+\cos\phi) \; .
\end{eqnarray}

The eigenstates of $H$ are formed by populating the vacuum states $\ket{0^\pm}$ of Bogoliubov fermions $a_q,\,q\in\Gamma^\pm$, and by taking care of the parity requirements in \eqref{supp_Hamiltonian_2}. 
The parity-dependent vacuum states are given by
\begin{equation}
\label{Bogoliubov vacuum}
\ket{0^\pm}=\prod\limits_{0<q<\pi,\; q\in\Gamma^\pm} \big(\cos\theta_q-\imath\sin\theta_q \ b_q^\dagger b_{-q}^\dagger \big) \ket{0},
\end{equation}
where $\ket{0}\equiv\bigotimes_{j=1}^N\ket{\uparrow_j}$ is the vacuum for Jordan-Wigner fermions, satisfying the relation $c_j\ket{0}=0 \; \forall j$.
As it is easy to see from eq.~\eqref{Bogoliubov vacuum}, the vacuum states $\ket{0^+}$ and $\ket{0^-}$ by construction have even $\Pi^z$ parity.
Since each excitation $a_q^\dagger$ changes the parity of the state it follows that the eigenstates of $H$ belonging to $\Pi^z=-1$ sector are of the form $a_{q_1}^\dagger a_{q_2}^\dagger...a_{q_{m}}^\dagger\ket{0^-}$ with $q_i\in\Gamma^-$ and $m$ odd, while $\Pi^z=+1$ sector eigenstates are of the same form but with $q_i\in\Gamma^+$, $m$ even and the vacuum $\ket{0^+}$ used. 

On the other hand, as we have discussed in the previous section of these supplementary materials, from an eigenstate of one parity of $\Pi^z$ we can, by applying $\Pi^x$, obtain a second eigenstate, with the same energy, but different $\Pi^z$ parity.
This implies that to each aforementioned odd parity state, for instance, there is a corresponding even parity state $\Pi^x a_{q_1}^\dagger a_{q_2}^\dagger...a_{q_{m}}^\dagger\ket{0^-}$ with the same energy. 

In accordance with these facts, and keeping in mind that, as we can see from eq.~\eqref{supp_energy_momenta}, in the range of $\phi$ of our interest there is no momenta in the odd sector with a negative energy, the ground states in the odd parity sector of $\Pi_z$ are constructed by exciting the lowest energy modes $q\in\Gamma^-$ and have the form $a_q^\dagger\ket{0^-}$.
To each such state is associated an equivalent ground state in the even sector of the form $\Pi^x a_q^\dagger\ket{0^-}$. 
Similarly, the lowest lying excited states are obtained by exciting the other single modes. Therefore, the ground state is part of a band of $2N$ state, in which the energy gap between the states is, due to the spectrum of the form eq.~\eqref{supp_energy_momenta}, closing algebraically with the system size. The closing of the gap is a phenomenology analogous to Refs.~\cite{Dong2016,Giampaolo2019,Maric2019}, and is an aspect of geometrical frustration in general.

In the region $\phi\in(-\pi/4,0)$, studied in Ref.~\cite{Maric2019}, the energy in eq.~\eqref{supp_energy_momenta} for the momenta in the odd sector is minimized by $q=0$. So the ground state manifold is two-fold degenerate, spanned by the states $a_0^\dagger\ket{0^-}$ and $\Pi^xa_0^\dagger\ket{0^-}$. 
On the other hand, for $\phi\in(0,\pi/4)$ the energy would be minimized assuming $q=\pm\pi/2$. 
However, for any finite system with odd $N$ the momenta $q=\pm\pi/2$ are not allowed.
As a consequence the modes in the odd sector with the lowest energy, that we denote as $\pm p\in\Gamma^-$, are given by
\begin{equation}\label{momentum p}
p=
\begin{cases}
&\frac{\pi}{2}\left( 1- \frac{1}{N} \right) \; , \quad N\textrm{ mod }4=1\\
&\frac{\pi}{2}\left( 1+ \frac{1}{N} \right) \; , \quad N\textrm{ mod }4=3
\end{cases}
\end{equation}
Hence the two states $\ket{\pm p}= a_{\pm p}^\dagger\ket{0^-}$ represent the two ground states in the odd parity sector. The ground state manifold is, therefore, four-fold degenerate and a generic ground state can be written as a superposition
\begin{equation}\label{ground state superposition}
\ket{g}=u_1\ket{p}+u_2\ket{-p}+u_3 \ \Pi^x\ket{-p}+u_4 \ \Pi^x\ket{p} \; ,
\end{equation}
where we have assumed that the normalization condition $\sum_i\left|u\right|^2=1$ is satisfied.

\section{The Translation Operator}

The lattice translation operator $T$ is a linear operator that shifts cyclically all the spins in the lattice by one site. 
To define it, we choose a basis of the space and specify its action on the basis.
One basis of the Hilbert space of $N$ spins are the states
\begin{equation}\label{basis}
\ket{\psi}=\bigotimes_{k=1}^N(\sigma_k^-)^{n_k}\ket{\uparrow_k} \ ,
\end{equation}
where $n_1,n_2,...,n_N\in\{0,1\}$. 
The translation operator $T$ can then be defined by
\begin{equation}
\label{trasl_op_1}
T \ket{\psi}=\bigotimes_{k=1}^N(\sigma_k^-)^{n_{k+1}}\ket{\uparrow_k} \ , 
\end{equation}
where we make the identification $n_{N+1}\equiv n_{1}$.
From eq.~\eqref{trasl_op_1} it follows immediately that, for each state $\ket{\psi}$, we have that $\bra{\psi}T^\dagger T\ket{\psi}=1$. 
Hence the translation operator is unitary, i.e. $T^\dagger T=\mathbb{1}$
and the adjoint $T^\dagger$ plays the role of the translation operator in the other direction.
Moreover, applying the $T$ operator $N$ times translates the spins by the whole lattice and results in recovering the initial state, implying the idempotence of order $N$ of $T$, i.e. $T^N=\mathbb{1}$.
As a consequence, the only possible eigenvalues of the translation operator are the $N$-th roots of unity, given by $e^{\imath q}, q\in\Gamma^-$.

On the other hand, moving from the spin states to the operators, it is easy to see that the translation operator shifts the Pauli operators as
\begin{equation}\label{T spins}
T^\dagger\sigma_j^\alpha T=\sigma_{j+1}^\alpha , \quad \alpha=x,y,z \; ,
\end{equation}
where $\sigma_{N+1}^\alpha=\sigma_1^\alpha$, and, consequently it commutes with both the Hamiltonian in eq.~\eqref{Hamiltonian}  ($[T,H]=0$) and the parity operators \mbox{($[T,\Pi^\alpha]=0$ for $\alpha=x,y,z$)}.

The fact that the Hamiltonian and the translation operator commute implies that they admit a complete set of common eigenstates. In the following we prove that such a complete set is made by the eigenstates introduced in the previous section. Let us start by proving the following theorem.

\begin{thm}\label{tm 1}
	\begin{minipage}[t]{0.85\columnwidth}
		\begin{enumerate}[label=(\alph*)]
			\item The states $b_{q_1}^\dagger b_{q_2}^\dagger...b_{q_m}^\dagger\ket{0}$, with $m$ odd and $\{q_k\}\subset\Gamma^-$, are eigenstates of $T$ with eigenvalue equal to $\exp\big[\imath \sum_{k=1}^m q_k \big]$. \label{part a}
			\item The states $b_{q_1}^\dagger b_{q_2}^\dagger...b_{q_m}^\dagger\ket{0}$, with $m$ even and $\{q_k\}\subset\Gamma^+$, are eigenstates of $T$ with eigenvalue equal to $\exp\big[\imath \sum_{k=1}^m q_k \big]$. \label{part b}
		\end{enumerate}
	\end{minipage}
\end{thm}

\begin{proof}
	
	We write $\prod_{k=1}^m b_{q_k}^\dagger$ to indicate the ordered product of fermionic operators $b_{q_1}^\dagger b_{q_2}^\dagger...b_{q_m}^\dagger$. From the defining properties of $T$ we know how it acts on spin states and how it transforms the spin operators. Hence to study its action on the fermionic states $\left(\prod_{k=1}^m b_{q_k}^\dagger\right)\ket{0}$ it is convenient to write them in terms of spin states.
	This can be done in two steps.
	At first, using the eq.~\eqref{fourier transformed JW fermions}, we can write our state in terms of the Jordan-Wigner fermions, obtaining
	\begin{equation}\label{psi q 1}
	\!\!\!\!\left(\prod_{k=1}^m b_{q_k}^\dagger\right)\ket{0}=\frac{1}{N^{m/2}}
	\!\!\!\!\!\!\!\sum_{j_1,\ldots,j_m=1}^N\!\!\!\!\!\!\!
	e^{\imath \sum_{k=1}^m q_k j_k} \prod_{k=1}^m \left( c_{j_k}^\dagger \right)
	\ket{0} \ .
	\end{equation}
	Being the $c_{j_k}^\dagger$ operators fermionic, only the terms with all different $j_k$ survive.
	The second step is to invert the Jordan-Wigner mapping to bring back the fermionic states to spin ones.
	To do this step we first sort the fermionic operators, after which it's easy to invert the Jordan-Wigner transformation. To provide an example we have
	\begin{equation}
	c_1^\dagger c_4^\dagger c_2^\dagger\ket{0}=-c_1^\dagger c_2^\dagger c_4^\dagger \ket{0} = -\sigma_1^-   (\sigma_1^z)\sigma_2^-  (\sigma_1^z \sigma_2^z \sigma_3^z) \sigma_4^- \bigotimes_{k=1}^N\ket{\uparrow_k} = -\sigma_1^-   \sigma_2^-  \sigma_4^- \bigotimes_{k=1}^N\ket{\uparrow_k} \ .
	\end{equation}
	More generally we can write
	\begin{equation}\label{fermions spins}
	\bigotimes_{k=1}^m \left( c_{j_k}^\dagger \right) \ket{0}=
	S[\{j_k\}] \bigotimes_{k=1}^m  \left( \sigma_{j_k}^- \right) \bigotimes_{k'=1}^N  \ket{\uparrow_{k'}} \ ,
	\end{equation}
	where $S[\{j_k\}] $ is the sign of the permutation that brings the tuple $\{j_k\}$ to normal order. 
	Hence, the states \eqref{psi q 1} can be re-written in terms of spin operators as
	\begin{eqnarray}
	\left(\prod_{k=1}^m b_{q_k}^\dagger\right)\ket{0}&=& \frac{1}{N^{m/2}}\sum_{j_1,\ldots,j_m=1}^N S[ (j_k)]
	e^{\imath \sum_{k=1}^m q_k j_k} \bigotimes_{k=1}^m  \left( \sigma_{j_k}^- \right) \bigotimes_{k'=1}^N  \ket{\uparrow_{k'}} \ . \nonumber
	\end{eqnarray}
	Having the representation of the state in terms of spins, it is easy to see what is the result of the application of $T$. 
	Using its discussed properties and taking into account that that $T$ leaves the state $\bigotimes_{k'=1}^N  \ket{\uparrow_k'}$ unchanged we recover
	\begin{eqnarray}
	T\left(\prod_{k=1}^m b_{q_k}^\dagger\right)\ket{0}&=& \frac{1}{N^{m/2}}\sum_{j_1,\ldots,j_m=1}^N S[\{j_k\}]
	e^{\imath \sum_{k=1}^m q_k j_k} \bigotimes_{k=1}^m  \left( \sigma_{j_k-1}^- \right) \bigotimes_{k'=1}^N  \ket{\uparrow_{k'}}\nonumber\\
	&=& \frac{e^{\imath \sum_{k=1}^m q_k}}{N^{m/2}}\sum_{j_1,\ldots,j_m=1}^N S[\{j_k\}]
	e^{\imath \sum_{k=1}^m q_k (j_k-1)} \bigotimes_{k=1}^m  \left( \sigma_{j_k-1}^- \right) \bigotimes_{k'=1}^N  \ket{\uparrow_{k'}} \ .
	\end{eqnarray}
	
	Let us focus now on part \ref{part a} of the Theorem. We have two different cases. If none of the elements in $\{j_k\}$ is equal to $1$, then none of the elements in $\{j_k-1\}$ is equal to zero, and trivially $S[\{j_k\}]=S[\{j_k-1\}]$. On the contrary if one element of $\{j_k\}$ is equal to 1, then $j_k-1$ becomes 0. 
	However, the number $m$ of the elements in $\{j_k\}$ is odd. Hence to move an element from the first to the last place requires an even number $m-1$ of permutations and hence the sign of the permutation $S[\{j_k\}]=S[\{j_k-1\}]$ remains the same if we replace $j_k-1=0$ with $N$. From this and the fact that, since $\{q_k\}\subset\Gamma^-$, the exponential $e^{\imath q_k(j_k-1)}$ remains the same if we replace $j_k-1=0$ with $N$, it follows that we can write
	\begin{eqnarray}
	\label{T psi 1}
	T\left(\prod_{k=1}^m b_{q_k}^\dagger\right)\ket{0}
	&=& \frac{e^{\imath \sum_{k=1}^m q_k}}{N^{m/2}}\sum_{j_1,\ldots,j_m=1}^N S[\{j_k-1\}]
	e^{\imath \sum_{k=1}^m q_k (j_k-1)} \bigotimes_{k=1}^m  \left( \sigma_{j_k-1}^- \right) \bigotimes_{k'=1}^N  \ket{\uparrow_{k'}} \ ,
	\end{eqnarray}
	where, if for some $k$ we have $j_k-1=0$, we can identify it with $j_k-1=N$. Because of this identification it's easy to write each term in the sum in terms of fermions:
	\begin{eqnarray}
	\label{T psi 2}
	T\left(\prod_{k=1}^m b_{q_k}^\dagger\right)\ket{0}
	&=& \frac{e^{\imath \sum_{k=1}^m q_k}}{N^{m/2}}\sum_{j_1,\ldots,j_m=1}^N 
	e^{\imath \sum_{k=1}^m q_k (j_k-1)} \prod_{k=1}^m \left( c_{j_k-1}^\dagger \right)  \ket{0} \ .
	\end{eqnarray}
	In eq.~\eqref{T psi 2} we can, again on the basis of the identification of $0$ with $N$, rename the indices to get
	\begin{eqnarray}
	\label{T acting on psi 0}
	T\left(\prod_{k=1}^m b_{q_k}^\dagger\right)\ket{0}
	&=& \frac{e^{\imath \sum_{k=1}^m q_k}}{N^{m/2}}\sum_{j_1,\ldots,j_m=1}^N 
	e^{\imath \sum_{k=1}^m q_k j_k} \prod_{k=1}^m \left( c_{j_k}^\dagger \right)  \ket{0}=\exp\left(\imath \sum_{k=1}^m q_k\right)\left(\prod_{k=1}^m b_{q_k}^\dagger\right)\ket{0} \ ,
	\end{eqnarray}
	which proves part \ref{part a} of Theorem \ref{tm 1}. Part \ref{part b} is proven in a similar way.
\end{proof}

From Theorem \ref{tm 1} it follows immediately, by taking into account the definition of the Bogoliubov particles in eq.~\eqref{Bogoliubov particles}, the definition of the Bogoliubov vacua in eq.~\eqref{Bogoliubov vacuum}, and the linearity of the translation operator, that also the Hamiltonian eigenstates $\left(\prod_{k=1}^m a_{q_k}^\dagger\right)\ket{0^\pm}$ are eigenstates of $T$ with eigenvalues equal to $\exp\left(\imath \sum_{k=1}^m q_k\right)$.

\section{The Mirror Operator}

As we have seen in the first section of these supplementary materials, the Hamiltonian is invariant under the mirror transformation with respect to a generic site $k$ that changes spin operators defined on the site $j$ to ones defined on the site $2k-j$. Note that, with the odd number N of sites we work with, in a circular geometry, the line of mirror reflection crosses a site and a bond. Hence, only site $k$ remains unchanged by the mirror action.

As we have done for translations, the mirror transformation can also be expressed by the action of a suitable operator. The mirror operator $M_k$, that makes the mirror transformation of the states with respect to the $k$-th site, is defined by its action on the spin basis states $\ket{\psi}$, defined in eq.~\eqref{basis}, as
\begin{equation}
\label{definition_M}
M_k \ket{\psi}=M_k \bigotimes_{j=1}^N(\sigma_j^-)^{n_j}\ket{\uparrow_j}=
\bigotimes_{j=1}^N(\sigma_j^-)^{n_{2k-j}} \ket{\uparrow_j} \ ,
\end{equation}
where, as always, $n_{j+N}\equiv n_j$. From eq.~\eqref{definition_M} it follows immediately that, for each state $\ket{\psi}$, we have that $\bra{\psi} M_k^\dagger M_k \ket{\psi}=1$. 
Hence, as the translation operator, also $M_k$ is unitary, i.e. $M_k^\dagger M_k=\mathbb{1}$. 
Moreover, applying the mirror operator two times results in recovering the initial state, hence implying the idempotence of order $2$ of the operator $M_k$, i.e. $M_k^2=\mathbb{1}$. 
This implies that $M_k$ is also Hermitian, i.e. $M_k^\dagger=M_k$, and that the only possible eigenvalues of $M_k$ are $\pm 1$.
Moreover, different mirror operators are related by translations, 
\begin{equation}\label{mirror translation}
T^\dagger M_k T=M_{k+1}
\end{equation}
From this relation it is also clear that the mirror operators do not commute with the translation operator ($[M_k,T]\neq0$).

Since each of the mirror operators commutes with the Hamiltonian, the Hamiltonian shares a common basis with each one of them. 
The following theorem gives the relation between the eigenstates we have constructed and the mirror operators. 
Essentially, the mirror operators change the sign of the momenta of the excitations, up to a possible phase factor, depending on $k$. 
Since different mirror operators are related by eq.~\eqref{mirror translation} we focus on the one with $k=N$ for which the phase factor is absent.
\begin{thm}	\label{tm 2}
	\begin{minipage}[t]{0.85\columnwidth}
		\begin{enumerate}[label=(\alph*)]
			\item 	The mirror operator $M_N$ acts on the states $b_{q_1}^\dagger b_{q_2}^\dagger...b_{q_m}^\dagger\ket{0}$, with $m$ odd and $\{q_k\}\subset\Gamma^-$, as
			\begin{equation}
			M_N \ b_{q_1}^\dagger b_{q_2}^\dagger...b_{q_m}^\dagger\ket{0} = b_{-q_m}^\dagger b_{-q_{m-1}}^\dagger...b_{-q_1}^\dagger\ket{0}.
			\end{equation}
			\item The mirror operator $M_N$ acts on the states $b_{q_1}^\dagger b_{q_2}^\dagger...b_{q_m}^\dagger\ket{0}$, with $m$ even and $\{q_k\}\subset\Gamma^+$, as
			\begin{equation}
			M_N \ b_{q_1}^\dagger b_{q_2}^\dagger...b_{q_m}^\dagger\ket{0} = b_{-q_m}^\dagger b_{-q_{m-1}}^\dagger...b_{-q_1}^\dagger\ket{0}.
			\end{equation}
		\end{enumerate}
	\end{minipage}
\end{thm}


The theorem is proven in a similar way as Theorem \ref{tm 1}, and we omit the details. 
The other mirror operators $M_k$, with $k\neq N$, would introduce an additional phase factor by acting on the aforementioned eigenstates. 
The phase factor depends on the momentum of the state and can be reconstructed from eq.~\eqref{mirror translation}. 
The $N$-th site being special here is a consequence of its special position in the Jordan-Wigner transformation, which implicitly enters in the definition of the states we work on. 
In the proof of Theorem \ref{tm 2} the $N$-th site is special because for $k=N$ the exponentials of the type $e^{\imath qj}$ can be replaced by $e^{\imath (-q)(2k-j) }$, while for other $k$ a compensating factor has to be introduced. 

Similarly as after Theorem \ref{tm 1}, but using also the property $\theta_{-q}=-\theta_q$ of the Bogoliubov angle, it follows from Theorem \ref{tm 2} that the mirror operator $M_N$ acts on the Hamiltonian eigenstates $a_{q_1}^\dagger a_{q_2}^\dagger...a_{q_m}^\dagger\ket{0^\pm}$ as $M_N \ a_{q_1}^\dagger a_{q_2}^\dagger...a_{q_m}^\dagger\ket{0^\pm} = a_{-q_m}^\dagger a_{-q_{m-1}}^\dagger...a_{-q_1}^\dagger\ket{0^\pm}$. Note that, as a consequence, only the states with the total momentum satisfying $\exp\big[\imath \sum_{j=1}^m q_j\big] = \pm 1$ can simultaneously be the eigenstates of $T$ and $M_N$. Finally, let us notice that mirroring does not change the parity and so the mirror operator commutes with the parity operators, i.e. $[M_N,\Pi^\alpha]=0 , \ \alpha=x,y,z.$


\section{The Spatial Dependence of the Magnetization}

As we have proved in the section about the exact solution of the model, in the region $\phi\in(0,\pi/4)$ the ground state manifold is four fold degenerate. 
Hence a large variety of possible ground states with different magnetic properties can be selected. 
Among them, the ground states at the center of the manuscript to which this supplementary material is attached are of the form
\begin{equation}\label{state Incommensurate}
\ket{\tilde{g}}=\frac{1}{\sqrt{2}}\big(\ket{p}+e^{\imath \theta} \ \Pi^x\ket{-p} \big),
\end{equation}
where $\theta$ is a free phase. 
For such state the magnetization in the $\gamma$ direction, with $\gamma=x,\,y$, shows the peculiar incommensurate antiferromagnetic order that we discussed in the main paper and that we will elaborate on in the following.
By definition, the magnetization in the $\gamma$ direction is equal to
\begin{equation}
\label{step_1_mag}
\braket{\sigma_j^\gamma}_{\tilde{g}}=\frac{1}{2} \big( e^{\imath \theta}\bra{p}\sigma_j^\gamma \Pi^x\ket{-p}+\textrm{c.c.} \big) \ .
\end{equation}
The magnetization is thus determined by the quantities $\bra{p}\sigma_j^\gamma\Pi^x\ket{-p}$, which are matrix elements of the spin string operators $\sigma_j^\gamma$ between the ground states vectors $\ket{p}$ and $\Pi^x\ket{-p}$. The matrix elements at any site $j$ can be related to the ones at site $N$, using the translation operator. Using the relation $\sigma_k^\alpha=(T^\dagger)^{k}\sigma_N^\alpha(T)^{k}$ and knowing the eigenvalues of $T$ we get
\begin{equation}
\bra{p}\sigma_j^\gamma \Pi^x \ket{-p}=e^{-\imath 2pj} \bra{p} \sigma_N^\gamma \Pi^x \ket{-p} \ .
\end{equation}
The advantage of expressing the quantity $\bra{p}\sigma_j^\gamma\Pi^x\ket{-p}$ in terms of the one at site $j=N$ is that this last one is real for $\gamma=x$ and purely imaginary for $\gamma=y$, as we will now show. The reason why the $N$-th site is special is because the Jordan-Wigner transformation, which implicitly enters into the definition of the states, breaks the invariance under spatial translation by identifying a first (and a last) spin in the ring. 

To show that the quantity $\bra{p}\sigma_N^x\Pi^x\ket{-p}$ is real we resort to the mirror operator, which relates the states with opposite momentum as $M_N\ket{p}=\ket{-p}$, according to Theorem \ref{tm 2}. Using this relation and taking into account that $M_N$ is hermitian we get
\begin{equation}
\bra{p}\sigma_N^x \Pi^x\ket{-p}= \bra{-p}M_N\sigma_N^x\Pi^x M_N\ket{p} \ .
\end{equation}
But, as we have said, $\Pi^x$ commutes with the mirror operator, which together with the property $M_N \sigma_N^x M_N =\sigma_N^x$ gives
\begin{equation}
\bra{p}\sigma_N^x \Pi^x\ket{-p}=\bra{-p}\sigma_N^x \Pi^x\ket{p}=\big(\bra{p}\sigma_N^x\Pi^x\ket{-p}\big)^* \ ,
\end{equation}
where the last equality holds because the operator $\sigma_N^x \Pi^x$ is hermitian. Hence $\bra{p}\sigma_N^x\Pi^x\ket{-p}$ is equal to its conjugate and therefore real. To show that $\bra{p}\sigma_N^y \Pi^x\ket{-p}$ is purely imaginary we can use the same method together with the property that $\sigma_N^y\Pi^x$ is antihermitian, or we can use the relation
\begin{equation}
\Pi^x=(-\imath)^N\Pi^y\Pi^z
\end{equation}
and the eigenstate property $\Pi^z\ket{\pm p}=-\ket{\pm p}$, which give
\begin{equation}
\bra{p}\sigma_N^y \Pi^x\ket{-p}=-(-\imath)^N\bra{p}\sigma_N^y\Pi^y\ket{-p}.
\end{equation}
The quantity $\bra{p}\sigma_N^y\Pi^y\ket{-p}$ is real, by the same argument which shows that $\bra{p}\sigma_N^x\Pi^x\ket{-p}$ is real and the factor in front, due to oddity of $N$, makes the whole quantity imaginary.

Taking these properties into account, we get the following spatial dependence for the magnetizations
\begin{align}
\braket{\sigma_j^x}_{\tilde{g}}&=\cos(2pj-\theta)\bra{p}\sigma_N^x \Pi^x\ket{-p} \ , \\
\braket{\sigma_j^y}_{\tilde{g}}&=\cos(2pj-\theta+N\frac{\pi}{2}+\pi)\bra{p}\sigma_N^y \Pi^y\ket{-p} \ .
\end{align}
Inserting the exact value of the momentum \eqref{momentum p}, which is equal to $p=\frac{\pi}{2}+(-1)^{\frac{N+1}{2}}\frac{\pi}{2N}$, we get finally the dependence of the magnetizations on the position in the ring,
\begin{equation}\label{incommensurate magnetization}
\braket{\sigma_j^\gamma}_{\tilde{g}}=(-1)^j\cos\left[\pi\frac{j}{N}+\lambda(\gamma,\theta,N)\right]\bra{p}\sigma_N^\gamma \Pi^\gamma\ket{-p} \ ,
\end{equation}
where
\begin{equation}
\lambda(\gamma,\theta,N) \equiv
\begin{cases}
(-1)^{\frac{N-1}{2}}\theta , \quad & \gamma=x\\
(-1)^{\frac{N-1}{2}}\theta +\frac{\pi}{2}, \quad & \gamma=y 
\end{cases}.
\end{equation}

The magnetization is antiferromagnetic, i.e. staggered, but its magnitude is modulated. Since the number of sites is odd, it is not possible to have every bond aligned antiferromagnetically, but there is necessarily at least a one ferromagnetic one. The magnetization is modulated in such a way to achieve the minimal absolute value at the ferromagnetic bond, thus minimizing the energy. 
The position of this ferromagnetic bond is determined by the phase $\theta$. The position of the ferromagnetic bond of the magnetization in the $x$ direction is shifted by half of the ring with the respect to the ferromagnetic bond of the magnetization in the $y$ direction.

\section{Explicit evaluation of the magnetizations on the $N$-th site}

We can evaluate the magnetization on the $N$-th spin of the lattice exploiting a method similar to the one we developed in Ref.~\cite{Maric2019}. It consists on expressing the matrix elements $\bra{p}\sigma_N^\gamma\Pi^x\ket{-p}$ in terms of expectation values of $\sigma_N^\gamma\Pi^x$ in a definite $\Pi^z$ parity state, using the representation of $\sigma_N^\gamma\Pi^x$ in terms of Majorana fermions
\begin{equation}\label{Majorana fermions definition}
A_j= \Big(\bigotimes\limits_{l=1}^{j-1}\sigma_l^z \Big)\otimes \sigma_j^x \; , \quad B_j= \Big(\bigotimes\limits_{l=1}^{j-1} \sigma_l^z \Big) \otimes\sigma_j^y \; ,
\end{equation}
using Wick's theorem to express the expectation values as a determinant, and finally evaluating the determinant.

We express $\bra{p}\sigma_N^\gamma\Pi^x\ket{-p}$ in terms of expectation values of $\sigma_N^\gamma\Pi^x$ on ground states living in the odd parity sector of $\Pi^z$. 
A general ground state belonging to the odd parity sector of $\Pi^z$ can be written as in eq.~\eqref{ground state superposition} setting $u_3=u_4=0$,
\begin{equation}\label{states single parity}
\ket{u_1,u_2}\equiv u_1\ket{p}+u_2\ket{-p}
\end{equation}
It is immediate to see that
\begin{equation}
\braket{\sigma_j^\gamma\Pi^x}_{u_1=\frac{1}{\sqrt{2}},u_2=\frac{1}{\sqrt{2}}}-\braket{\sigma_j^\gamma\Pi^x}_{u_1=\frac{1}{\sqrt{2}},u_2=-\frac{1}{\sqrt{2}}}=\bra{p}\sigma_j^\gamma\Pi^x\ket{-p}+\bra{-p}\sigma_j^\gamma\Pi^x \ket{p}
\end{equation}
Using the properties of the mirror operator, in the previous section we have shown that $\bra{p}\sigma_N^x \Pi^x \ket{-p}=\bra{-p}\sigma_N^x\Pi^x\ket{p}$, while in an analogous way we have also $\bra{p} \sigma_N^y\Pi^x\ket{-p}=\bra{-p} \sigma_N^y\Pi^x\ket{p}$. Using these relations we get, finally,
\begin{equation}
\bra{p}\sigma_N^\gamma\Pi^x\ket{-p}=\frac{1}{2}\Big(\braket{\sigma_N^\gamma\Pi^x}_{u_1=\frac{1}{\sqrt{2}},u_2=\frac{1}{\sqrt{2}}}-\braket{\sigma_N^\gamma\Pi^x}_{u_1=\frac{1}{\sqrt{2}},u_2=-\frac{1}{\sqrt{2}}} \Big) . \label{FF2 to compute}
\end{equation}

Now, $\sigma_N^x\Pi^x$ and $\sigma_N^y\Pi^x$, being products of spin operators, can be expressed in terms of Majorana fermions, as
\begin{equation}\label{product spins Majoranas}
\sigma_N^x\Pi^x=(-1)^{\frac{N-1}{2}}\prod\limits_{l=1}^{\frac{N-1}{2}} (-\imath A_{2l}B_{2l-1}), \qquad\sigma_N^y\Pi^x=-\imath (-1)^{\frac{N-1}{2}}\left(\prod\limits_{l=1}^{\frac{N-1}{2}} (-\imath A_{2l}B_{2l-1})\right) (-\imath A_{N}B_{N}) \ .
\end{equation}
The expectation values of these operators in a definite $z$ parity ground state can be expressed as a Pfaffian of the matrix of two-point correlators, using Wick's theorem.

To do so, we write the state \eqref{states single parity} as a vacuum state for fermionic operators, in terms of which the Majorana fermions \eqref{Majorana fermions definition} are linear. These fermions are defined by
\begin{equation}
\alpha_p=u_1 a_p^\dagger+u_2a_{-p}^\dagger, \quad \alpha_{-p}=u_2 a_p-u_1a_{-p}
\end{equation}
and by $\alpha_q=a_q$ for $q\neq p, -p$. It's easy to check that the operators $\alpha_q$ satisfy fermionic anticommutation relations and annihilate the state $\ket{u_1,u_2}=\alpha_p\ket{0^-}$, i.e. that we have $\alpha_q\ket{u_1,u_2}=0$ for $q\in\Gamma^-$. Moreover, since the Majorana fermions in eq.~\eqref{Majorana fermions definition} can be written as a linear combination of Bogoliubov fermions $a_q,a_q^\dagger$, they can also be written as a linear combination of fermions $\alpha_q, \alpha_q^\dagger$. In this way, we are able to straightforwardly apply Wick's theorem to evaluate the string operators in eq.~\eqref{product spins Majoranas} over the chose ground state vectors.

The two-point correlators of Majorana fermions are evaluated to be
\begin{eqnarray}
\label{Majorana four fold AA}
\braket{A_jA_l}_{u_1,u_2}= \braket{B_jB_l}_{u_1,u_2}&=&\delta_{jl}
-\frac{2\imath}{N}(|u_1|^2-|u_2|^2)  \sin\big[p(j-l)\Big] , \\
\label{Majorana four fold AB}
-\imath\braket{A_jB_l}_{u_1,u_2}
&=& \frac{1}{N}\sum\limits_{q\in\Gamma^-}e^{\imath 2\theta_{q}}e^{-\imath p(j-l)} -\frac{2}{N} \cos\big[p(j-l)-2\theta_{p}\big] - \frac{2}{N}\Big(u_1^*u_2 \ e^{-\imath p (j+l)} +\textrm{c.c.} , \Big) 
\end{eqnarray}
where the Bogoliubov angle $\theta_q$ is defined in eq.~\eqref{arctan}. The Bogoliubov angle also satisfies
\begin{equation} \label{exp 2 theta}
e^{\imath 2\theta_{q}}=e^{\imath q} \frac{\cos\phi+\sin\phi \ e^{-\imath 2q}}{|\cos\phi+\sin\phi \ e^{-\imath 2q}|}.
\end{equation}
which should be used in eq.~\eqref{Majorana four fold AB} for the mode $q=0$, for which eq.~\eqref{arctan} is undefined.

As a matter of fact, in the evaluation of the matrix elements we encounter only states of the type $\left|u_1\right|=\left|u_2\right|=1/\sqrt{2}$, for which the correlators \eqref{Majorana four fold AA} 
vanish for $j\neq l$. 
This allows us to use the standard approach \cite{Lieb1961} on the basis of Wick's theorem to express the expectation value of \eqref{product spins Majoranas} as a determinant.
For $\braket{\sigma_N^y\Pi^x}_{u_1,u_2}$ we have that 
\begin{equation}
\braket{\sigma_N^y\Pi^x}_{u_1,u_2}=-\imath (-1)^{\frac{N-1}{2}}\det \mathbf{C} \ ,
\end{equation}
with the $(N+1)/2\times (N+1)/2$ correlation matrix $\mathbf{C}$ given by 
\begin{equation}
\mathbf{C}=
\left(
\begin{array}{cccccc}
F(2,1) & F(2,3) & F(2,5) & \cdots & F(2,N-2) & F(2,N) \\ 
F(4,1) & F(4,3) & F(4,5) & \cdots & F(4,N-2) & F(4,N) \\
\vdots & \vdots & \vdots & \ddots & \vdots & \vdots    \\
F(N-1,1) & F(N-1,3) & F(N-1,5) & \cdots & F(N-1,N-2) & F(N-1,N) \\
F(N,1) & F(N,3) & F(N,5) & \cdots & F(N,N-2) & F(N,N) \\
\end{array} 
\right) ,
\end{equation}
where $F(j,l)=-\imath\braket{A_{j}B_{l}}_{u_1,u_2}$. On the contrary 
\begin{equation}
\braket{\sigma_N^x\Pi^x}_{u_1,u_2} =(-1)^{\frac{N-1}{2}}\det \mathbf{C'} \ ,
\end{equation}
where the $(N-1)/2\times (N-1)/2$ correlation matrix $\mathbf{C'}$ is obtained from $\mathbf{C}$ by removing the last row and the last column.  

The determinants we encounter have a more complicated form than those for which the standard analytical approach \cite{Barouch1971} applies so we have evaluated them numerically.


\section{Perturbative analysis}

The main points of this work can also be seen from a simple perturbative analysis around the classical Ising point $\phi=0$. In addition, the perturbative analysis provides an analytical expressions for the matrix elements in the limit $\phi\to 0^+$.

At the classical Ising point $\phi=0$ the model is diagonal in the basis where $\sigma_j^x$ are diagonal. The ground state manifold is $2N$-fold degenerate and consists of kink states $\ket{j}$ and $\Pi^z\ket{j}$, for $j=1,2,...,N$. Here, the kink state $\ket{j}$ is defined as the state
\begin{equation}
\ket{j}=\ket{...,1,-1,1,1,-1,1,....} \ ,
\end{equation}
with the ferromagnetic bond $\sigma_j^x=\sigma_{j+1}^x=1$ between sites $j$ and $j+1$, and antiferromagnetic bonds between all other adjacent sites. The kink state $\Pi^z\ket{j}$, with all spins reversed, has the ferromagnetic bond $\sigma_j^x=\sigma_{j+1}^x=-1$ and all the other bonds antiferromagnetic. The parity of the states $\ket{j}$ is $\Pi^x=(-1)^{(N-1)/2}$, while $\Pi^z\ket{j}$ have, of course, the opposite parity. The higher energy states are separated by a finite gap and can be neglected in perturbation theory.

Increasing $\phi$ from zero to a small non-zero value the exact degeneracy between the kink states splits. The ground states, and the corresponding energies, are found by diagonalizing the perturbation $\sin\phi\sum_{j}\sigma_j^y\sigma_{j+1}^y$ in the basis of kink states. This has already been done in \cite{Maric2019} and details can be found there. It has been found that the ground states of the model in the limit $\phi\to 0$ are superpositions of kinks
\begin{equation}\label{supp_superposition_kinks}
\ket{s_q}=\frac{1}{\sqrt{N}} \sum_{j=1}^{N}e^{\imath q j}\ket{j}
\end{equation}
and $\Pi^z\ket{s_q}$, for $q\in\Gamma^-$. The corresponding energies are
\begin{equation}
E(q)=-(N-2)\cos\phi+2\sin\phi \cos (2q) \ .
\end{equation}

It's easy to see that for $\phi<0$ the energy is minimized by $q=0$, while for $\phi>0$ it is by $q=p$, where $p$ is given by eq.~\eqref{momentum p}, as in the exact solution. Evaluating the derivative of the ground state energy $E_g$ we find a discontinuity at $\phi=0$,
\begin{equation}
\frac{dE_g}{d\phi}\bigg|_{\phi\to0^-}-\frac{dE_g}{d\phi}\bigg|_{\phi\to0^+}=2\Big(1+\cos\frac{\pi}{N}\Big) ,
\end{equation}
which goes to a constant non-zero value in the thermodynamic limit $N\to\infty$. 

We now turn to the evaluation of the matrix element. We can identify the states from perturbation theory with those from the exact solution, in the limit $\phi\to0$, by looking at the eigenstates of various operators. The translation operator shifts the kink as $T\ket{j}=\ket{j-1}$, from which it follows that the states $\ket{s_q}$ are eigenstates of $T$ with the eigenvalue $e^{\imath q}$. The mirror operator acts on the kink states as $M_N\ket{j}=\ket{-j-1}$, and therefore on the superpositions as
\begin{equation}
M_N\ket{s_q}=e^{-\imath q}\ket{s_{-q}} \ .
\end{equation} 
Knowing that the eigenstates $\ket{q}$ from the exact solution have parity $\Pi^z=-1$, are eigenstates of $T$ with the eigenvalue $e^{\imath q}$ and that under mirroring behave as $M_N\ket{q}=\ket{-q}$ we can make the identification
\begin{equation}
\ket{q}=\frac{1-\Pi^z}{\sqrt{2}}\ket{s_q} \; , \quad \ket{-q}=\frac{1-\Pi^z}{\sqrt{2}}e^{-\imath q}\ket{s_{-q}} \ ,
\end{equation}
up to an irrelevant phase factor which is the same for the two states.

From the identification we can express the matrix elements as
\begin{align}
&\bra{q}\sigma_N^x\Pi^x\ket{-q}=(-1)^{\frac{N-1}{2}}e^{-\imath q}\bra{s_q}\sigma_j^x\ket{s_{-q}} \ ,
\end{align}
where the factor $(-1)^{(N-1)/2}$ stems from the parity of the states $\ket{s_q}$. Using the definition of the states on the right we get
\begin{align}
&\bra{q}\sigma_N^x\Pi^x\ket{-q}=(-1)^{\frac{N-1}{2}}\frac{1}{N}\sum\limits_{j=1}^{N}e^{-\imath 2qj}\bra{j}\sigma_N^x\ket{j} \ , 
\end{align}
which can be evaluated using the property of the kink states
\begin{equation}\label{supp_kink_states_magnetization}
\bra{j} \sigma_N^x \ket{j}=
\begin{cases}
(-1)^{j} , & j=1,2,...,N-1\\
1 , & j=N
\end{cases}
\end{equation}
that follows from their definition. We end up with 
\begin{align}
&\bra{q}\sigma_N^x\Pi^x\ket{-q}=(-1)^{\frac{N-1}{2}}\frac{1}{N \cos q} \ .
\end{align}

The matrix element for the ground state momentum $p=\pi/2+(-1)^{(N+1)/2}\pi/2N$ becomes 
\begin{equation}
\bra{p}\sigma_N^x\Pi^x\ket{-p}=\frac{1}{N\sin\frac{\pi}{2N}} \ ,
\end{equation}
and in the limit $\phi\to 0^+$ determines the maximum value the magnetization achieves over the ring in the ground state $\ket{\tilde{g}}$. For large $N$ it becomes
\begin{equation}
\bra{p}\sigma_N^x\Pi^x\ket{-p}=\frac{2}{\pi}+\frac{\pi}{12N^2}+O(N^{-4}) \ ,
\end{equation}
which approaches quadratically the value $2/\pi \approx 0.64$ in the thermodynamic limit.
\end{widetext}

\def\bibsection{\section*{\refname}}

\end{document}